\documentclass[11pt]{amsart}
\usepackage{amsmath,amssymb,subfigure,amsthm,mathtools,enumerate,setspace}\usepackage[mathscr]{eucal}

\usepackage{graphicx}
\usepackage{color}
\newtheorem{theorem}{Theorem}[section]
\newtheorem{lemma}[theorem]{Lemma}
\newtheorem{proposition}[theorem]{Proposition}
\newtheorem{definition}[theorem]{Definition}
\newtheorem{remark}[theorem]{Remark}

\newcommand{\A}{\mathcal{A}}
\newcommand{\Q}{\mathcal{Q}}
\newcommand{\R}{\mathbb{R}}

\newcommand{\Z}{\mathbb{Z}}
\newcommand{\N}{\mathbb{N}}
\newcommand{\G}{\mathcal{N}}
\newcommand{\E}{\mathbb{E}}
\renewcommand{\P}{\mathbb{P}}

\newcommand{\sd}{\Sigma\Delta}

\newcommand{\be}{\begin{enumerate}}
\newcommand{\bi}{\begin{itemize}}
\newcommand{\ee}{\end{enumerate}}
\newcommand{\ei}{\end{itemize}}

\title[Sigma-Delta quantization of sub-Gaussian frame expansions]{Sigma-Delta quantization of sub-Gaussian frame expansions and its application to compressed sensing}
\author{
Felix Krahmer
   , Rayan Saab
   , and {\"O}zg{\"u}r Y{\i}lmaz
}

\begin{document}
\maketitle

\begin{abstract}
  {Suppose that the collection $\{e_i\}_{i=1}^m$ forms a frame for $\R^k$,
  where each entry of the vector $e_i$ is a sub-Gaussian random variable. We consider expansions in such a frame, which are then quantized using a Sigma-Delta scheme.
  We show that an arbitrary signal in $\R^k$ can be recovered from its quantized frame coefficients up to an error which decays root-exponentially
  in the oversampling rate $m/k$. Here the quantization scheme is assumed to be chosen appropriately depending on the oversampling rate and the quantization alphabet can be coarse.
   The result holds with high
  probability on the draw of the frame uniformly for all signals. The crux of the argument is a bound on the
  extreme singular values of the product of a deterministic matrix and
  a sub-Gaussian frame. For fine quantization alphabets, we leverage this bound to show
  polynomial error decay in the context of compressed sensing. Our results extend 
  previous results for
  structured deterministic frame expansions and 
  Gaussian compressed sensing measurements.}
  {compressed sensing, quantization, random frames, root-exponential accuracy, Sigma-Delta,  sub-Gaussian matrices}
  \\
2010 Math Subject Classification: 94A12, 94A20, 41A25, 15B52

\end{abstract}
\section{Introduction}

\subsection{ Main problem}
In this paper we address the problem of digitizing, or
quantizing, generalized linear measurements of finite dimensional
signals. In this setting a signal is a vector $x \in \R^N$, and the
acquired measurements are inner products of $x$ with elements from a
collection of $m$ \emph{measurement vectors} in $\R^N$. This
generalized linear measurement model has received much attention
lately, both in the \emph{frame theory} literature where one considers $m\geq
N$,  e.g., \cite{casazza2013finite}, and in the \emph{compressed sensing} literature where $m
\ll N$, e.g., \cite{eldar2012compressed}.
For concreteness, let $\{a_i\}_{i=1}^m\subset \R^N$ denote the
measurement vectors. 
The generalized linear measurements are given by $y_i=\langle a_i,
x \rangle, \ i\in \{1,...,m\},$ and can be organized as a
 vector $y$, given by $y=Ax$. Note that here and
throughout the paper, we always consider column vectors. Our goal is
to \emph{quantize} the measurements, i.e., to map the components $y_i$
of $y$ to elements of a fixed finite set so that the measurements can
be stored and transmitted digitally. A natural requirement for such
maps is that they allow for accurate recovery of the underlying
signal. As the domain of the forward map is typically an
uncountable set and its range is finite, an exact inversion is, in
general, not possible. Accordingly we seek an approximate inverse,
which we shall refer to as a \emph{reconstruction scheme} or
\emph{decoder}. More precisely, let $\A\subset \R$ be a finite set,
which we shall refer to as the quantization alphabet, and consider a
set $\mathcal{X}$ of signals, which is compact in $\R^N$.  Then a
\emph{quantization scheme} is a map
$$
{\Q}: A\mathcal{X} \rightarrow \mathcal{A}^m
$$ 
and a \emph{reconstruction scheme} is of the form
$$
\Delta: \mathcal{A}^m \rightarrow \R^N.
$$ 
For $\Delta$ to be an approximate inverse of $\Q$, we need that the
reconstruction error \[\epsilon(x):=\| x- \Delta(\Q(Ax))\|_2\] viewed
as a function of $x$ is as small as possible in some appropriate
norm. Typical choices are $\|\epsilon\|_{L^\infty(\mathcal{X})}$, the
worst case error, and $\|\epsilon\|^2_{L^2(\mathcal{X})}$, the mean
square error.

We are interested in quantization and reconstruction schemes that
yield fast error decay as the number of measurements increases. Next,
we give a brief overview of the literature on quantization in both
frame quantization and compressed sensing settings and explain how
these two settings are intimately connected. Below, we use different
notations for the two settings to make it easier to distinguish between
them. 

\subsection{\bf Finite frame setting}
Let $\{e_i\}_{i=1}^m$ be a frame for $\R^k$, i.e., $m\ge k$, and assume that the
matrix $E$ whose $i$th row is $e_i^T$ has rank $k$, thus the
map $x\mapsto Ex$ is injective. Accordingly, one can reconstruct any
$x\in \R^k$ exactly from the frame coefficients $y=Ex$ using, e.g.,
any left inverse $F$ of $E$. As we explained above, the frame
coefficients $y$ can be considered as generalized measurements of $x$ and
our goal is to quantize $y$ such that the approximation error is
guaranteed to decrease as the number of measurements, i.e., $m$,
increases. 

\subsubsection{Memoryless scalar quantization}
The most naive (and intuitive) quantization method is rounding off
every frame coefficient to the nearest element of the quantizer
alphabet $\A$. This scheme is generally called memoryless scalar
quantization (MSQ) and yields (nearly) optimal performance when $m=k$
and $E$ is an orthonormal basis. However, as redundancy increases (say, we
keep $k$ fixed and increase $m$) MSQ becomes highly suboptimal. In
particular, one can prove that the (expected) approximation error via
MSQ with a given fixed alphabet $\A$ can never decay faster than
$(m/k)^{-1}$ regardless of the reconstruction scheme that is used and
regardless of the underlying $E$ \cite{GVT98}. This is significantly inferior
compared to another family of quantization schemes, called $\sd$
quantizers, where one can have an approximation error that decays like
$(m/k)^{-s}$ for any integer $s>0$, provided one uses an appropiate order
quantizer.  

\subsubsection{Sigma-Delta quantization} Despite its use in the
engineering community since the 1960's as an alternative quantization
scheme for digitizing band-limited signals (see, e.g.,
\cite{inose1963unity}), a rigorous mathematical analysis of $\sd$
quantization was not done until the work of Daubechies and Devore
\cite{daub-dev}. Since then, the mathematical literature on  $\sd$
quantization has grown rapidly. 

Early work on the mathematical theory of $\sd$ quantization has
focused on understanding the reconstruction accuracy as a function of
\emph{oversampling rate} in the context of \emph{bandlimited
  functions}, i.e., functions with compactly supported Fourier
transform. Daubechies and DeVore constructed in their
seminal paper \cite{daub-dev}  \emph{stable} $r$th-order
$\sd$ schemes with a one-bit alphabet. Furthermore, they proved that
when such an $r$th-order scheme is used to quantize an \emph{oversampled} bandlimited
function, the resulting approximation error is bounded by $C_r \lambda^{-r}$
where $\lambda>1$ is the \emph{oversampling ratio} and $C_r$ depends
on the fine properties of the underlying stable $\sd$ schemes.  For
a given oversampling rate $\lambda$ one can then optimize the order $r$ to
minimize the associated \emph{worst-case} approximation error, which,
in the case of the stable $\sd$ family of Daubechies and DeVore,
yields that the approximation error is of order $O(\lambda^{-c\log
  \lambda})$.

In \cite{G-exp}, G{\"u}nt{\"u}rk constructed an alternative infinite
family of $\sd$ quantizers of arbitrary order---refined later by Deift
et al. \cite{DGK10}---and showed that using these new quantizers one
can do significantly better. Specifically, using such
schemes (see Section~\ref{CoSD}) in the bandlimited setting, one
obtains an approximation error of order $O(2^{-c\lambda})$ where
$c=0.07653$ in \cite{G-exp} and $c\approx 0.102$ in
\cite{DGK10}.  In short, \emph{when quantizing bounded bandlimited
  functions}, one gets exponential accuracy in the oversampling rate
$\lambda$ by using these $\sd$ schemes. In other words, one can refine the approximation by increasing
the oversampling rate $\lambda$, i.e., by collecting more
measurements, \emph{exponentially} in the number of measurements without
changing the quantizer resolution. 
Exponential error decay rates are known to be optimal \cite{CD02, G-exp}; lower bounds for the constants $c$ for arbitrary coarse quantization schemes are derived in \cite{KW12}.
In contrast, it again follows from \cite{GVT98} that independently quantizing the sample values at best yields linear decay of the average approximation
error.

Motivated by the observation above that suggests that $\sd$ quantizers
utilize the redundancy of the underlying expansion effectively,
Benedetto et al.~\cite{benedetto2006sigma} showed that
$\sd$ quantization schemes provide a viable quantization scheme for
finite frame expansions in $\R^d$.  In particular,
\cite{benedetto2006sigma} considers $x\in \R^k$ with $\|x\|_2\le 1$ and
shows that the reconstruction error associated with the first-order
$\sd$ quantization (i.e., $r=1$) decays like $\lambda^{-1}$ as
$\lambda$ increases. Here the ``oversampling ratio'' $\lambda$ is
defined as $\lambda=m/k$ if the underlying frame for $\R^k$ consists
of $m$ vectors. This is analogous to the error bound in the case of
bandlimited functions with a first-order $\sd$ quantizer. Following
\cite{benedetto2006sigma}, there have been several results on $\sd$
quantization of finite frames expansions improving on the
$O(\lambda^{-1})$ approximation error by using higher order schemes,
specialized frames, and alternative reconstruction techniques, e.g.,
\cite{BP2,BPY2,LPY,blum:sdf,KSW12}. Two of these papers are of special
interest for the purposes of this paper: Blum et al.~showed in
\cite{blum:sdf} that frames with certain smoothness properties
(including harmonic frames) allow for the $\sd$ reconstruction error
to decay like $\lambda^{-r}$, provided alternative dual
frames---called Sobolev duals---are used for reconstruction. Soon
after, \cite{KSW12} showed that by using higher order $\sd$ schemes
whose order is optimally chosen as a function of the oversampling rate
$\lambda$, one obtains that the worst-case reconstruction error decays
like $e^{-C\sqrt{\frac{m}{k}}}$, at least in the case of two distinct
structured families of frames: harmonic frames and the so-called
\emph{Sobolev self-dual frames} which were constructed in
\cite{KSW12}.  For a more comprehensive review of $\sd$ schemes and
finite frames, see \cite{PSY12}. One of the fundamental contributions
of this paper is to extend these results to wide families of random
frames.

\subsubsection{Gaussian frames and quantization} Based on the above
mentioned results, one may surmise that structure and smoothness
are in some way critical properties of frames, needed for good error
decay in $\sd$ quantization. Using the error analysis techniques of
\cite{blum:sdf}, it can be seen, though, that what is critical for
good error decay in $\sd$ quantization is the ``smoothness'' and
``decay'' properties of the \emph{dual frame} that is used in the
reconstruction---see Section~\ref{sderranalysis},
cf. \cite{PSY12}. 
This observation is behind the seemingly surprising fact that Gaussian
frames, i.e., frames whose entries are drawn independently according
to $\G(0,1)$, allow for polynomial decay in $\lambda$ of the
reconstruction error \cite[Theorem A]{GLPSY}---specifically, one can
show that the approximation error associated with an $r$th-order $\sd$
scheme is of order $O(\lambda^{-r+\frac{1}{2}})$. The proof relies on
bounding the extreme singular values of the product of powers of a
deterministic matrix---the difference matrix $D$ defined in
Section~\ref{sderranalysis}---and a Gaussian random matrix. This
result holds uniformly with high probability on the draw of the
Gaussian frame. 

\subsection{Compressed sensing setting}
Compressed sensing is a novel paradigm in mathematical signal
processing that was spearheaded by the seminal works of Candes,
Romberg, Tao \cite{CRT05}, and of Donoho \cite{Donoho2006_CS}.
Compressed sensing is based on the observation that various classes of
signals such as audio and images admit approximately sparse
representations with respect to a known basis or frame. Central
results of the theory establish that such signals can be recovered
with high accuracy from a small number of appropriate, non-adaptive
linear measurements by means of computatinally tractable
reconstruction algorithms. 

More precisely, one considers $N$-dimensional, $k$-sparse signals,
i.e., vectors in the set 
$$\Sigma_k^N := \{z \in \R^N, |\rm{supp}(z)|\leq k \}.
$$ 
The generalized measurements are acquired via the $m\times N$ matrix
$\Phi$ where $k<m\ll N$. The goal is to recover $z\in\Sigma_k^N$ from
$y=\Phi z$.  In this paper, we focus on random matrices $\Phi$ with
independent sub-Gaussian entries in the sense of
Definitions~\ref{def:subgvar} and \ref{def:subgmat} below. It is
well-known \cite{ve12-1} that if $m>C k\log(N/k)$, where $C$ is an
absolute constant, with high probability, such a choice of $\Phi$
allows for the recovery of all $z\in\Sigma_k^N$ as the solution $z^\#$
of the $\ell_1$-minimization-problem
\begin{equation*}
 z^\#=\arg\min_{x} \|x\|_1 \quad \text{subject to } y=\Phi x.
 \end{equation*}

\subsubsection{Quantized compressed sensing}
The recovery guarantees for compressed sensing are provably stable
with respect to measurement errors. Consequently, this allows to
incorporate quantization into the theory, albeit naively, as one can
treat the quantization error as noise. The resulting error bounds for
quantized compressed sensing, however, are not satisfactory, mainly
because additional quantized measurements will, in general, not lead
to higher reconstruction accuracy \cite{GLPSY}. 
The fairly recent literature on quantization of compressed sensing
measurements mainly investigates two families of quantization methods:
the 1-bit or multibit memoryless scalar quantization (MSQ)
\cite{JHF11, JLBB13, PV11, ALPV13} and $\sd$ quantization of arbitrary
order \cite{GLPSY}.

The results on the MSQ scenario focus on replacing the naive
reconstruction approach outlined above by recovery algorithms that
exploit the structure of the quantization error. For Gaussian random
measurement matrices, it has been shown that approximate
reconstruction is possible via linear programming even if the
measurements are coarsely quantized by just a single bit
\cite{PV11}. For non-Gaussian measurements, counterexamples with
extremely sparse signals exist which show that, in general,
corresponding results do not hold \cite{PV11}. These extreme cases can
be controlled by introducing the $\ell_\infty$-norm of the signal as
an additional parameter, establishing recovery guarantees for
arbitrary sub-Gaussian random measurement matrices, provided that this
norm is not too large \cite{ALPV13}. All these results yield
approximations where the error does not decay faster than
$\lambda^{-1}$, where in this case the oversampling rate $\lambda$ is
defined as the ratio of $m$, the number of measurements, to $k$, the
sparsity level of the underlying signal. Again, it follows from \cite{GVT98} that
for independently quantized measurements, i.e., MSQ, no quantization
scheme and no recovery method can yield an error decay that is faster
than $\lambda^{-1}$.

This bottleneck can be overcome by considering $\sd$ quantizers, which
take into account the representations of previous measurements in each
quantization step \cite{GLPSY}.  The underlying observation is that
the compressed sensing measurements are in fact frame coefficients of
the sparse signal restricted to its support. 
Accordingly, the problem of quantizing compressed sensing measurements
is a frame quantization problem, even though the ``quantizer'' does
not know what the underlying frame is. This motivates a two-stage
approach for signal recovery:

In a first approximation, the quantization error is just treated as
noise, a standard reconstruction algorithm is applied, and the indices
of the largest coefficients are retained as a support estimate. 
Once the support has been identified, the measurements carry redundant
information about the signal, which is exploited in the quantization
procedure by applying frame quantization techniques. 

This two-stage approach has been analyzed in \cite{GLPSY} for the specific case of Gaussian random measurement matrices. In particular, it was shown that under mild size assumptions on the non-zero entries of the underlying sparse signals, this approach can be carried through. Consequently, in the case of Gaussian measurement matrices, one obtains that the approximation error associated with an $r$th-order $\sd$ quantizer is of order $O(\lambda^{\alpha(-r+\frac{1}{2}}))$ where $\alpha\in (0,1)$---see \cite[Theorem B]{GLPSY}. These results hold uniformly with high probability on the draw of the Gaussian measurement matrix provided $m \ge C k (\log N)^{1/(1-\alpha)}$.

\subsection{Contributions}
Our work in this paper builds on these results, and generalizes
them. Our contributions are two-fold. On the one hand, we establish
corresponding results in the \emph{compressed sensing setting} which
allow arbitrary, independent, fixed variance \emph{sub-Gaussian} (in
the sense of Definition~\ref{def:subgvar} below) random variables as
measurement matrix entries. In particular, this includes the important
case of Bernoulli matrices, whose entries are renormalized independent
random signs. More precisely, in Theorem~\ref{thm:polynomial} we
prove a refined version of the following.

\begin{theorem}\label{thm:I_intro}
Let $\Phi$ be an $m \times N$ matrix whose entries are appropriately normalized independent sub-Gaussian random variables  and suppose that $\lambda:=m/k \geq \Big(C  \log(eN/k)\Big)^\frac{1}{1-\alpha}$
 where  $\alpha\in(0,1)$. With high probability the $r$-th order $\sd$ reconstruction $\hat{z}$ satisfies
 $$\|z-\hat{z}\|_2 \leq C \lambda^{-\alpha (r-1/2)}\delta,$$ 
 for all $z\in \Sigma_k^{N}$ for which $\min\limits_{j\in
   \rm{supp}{(z)}}|z_j| > C \delta$.  Here $\delta$ is the resolution of the $\sd$ quantization alphabet and $C$ is an appropriate
 constant that depends only on $r$.
\end{theorem}

 Our second line of contributions is on frame quantization: We show that using appropriate $\sd$ quantization schemes, we obtain root-exponential decay of the reconstruction error with both Gaussian and sub-Gaussian frame entries.  In particular, in Theorem~\ref{thm:root-exponential} we prove a refined version of the following result.  


\begin{theorem}\label{thm:II_intro}
Let $E$ be an $m \times k$ matrix whose entries are appropriately normalized independent sub-Gaussian random variables. Suppose that $\lambda:=m/k$ satisfies $\lambda \geq \lambda_0$, where $\lambda_0$ is a constant independent of $k$ and $m$.  Then with high probability  on the draw of $E$, the corresponding reconstruction $\hat{x}$ from a $\sd$ scheme of appropriate order satisfies 
$$\|x-\hat{x}\|_2  \leq C e^{-c \sqrt{\lambda}}.$$ 
 for all $x$ with $\|x\|_2\leq 1$. Here $c,C$ are appropriate constants.
\end{theorem}

Note that a key element of our proof, which may be of independent interest, pertains to the extreme singular values of the product of a deterministic matrix with quickly decaying singular values and a sub-Gaussian matrix, see Proposition~\ref{thm:main_1}.

\begin{remark}
All of the constants in the above theorems can be made explicit. Moreover, the quantization schemes are explicit and tractable, as are the reconstruction algorithms; however, the quantization scheme and reconstruction algorithms are different between Theorems \ref{thm:I_intro}  and \ref{thm:II_intro}. Please see Theorems \ref{thm:root-exponential} and \ref{thm:polynomial} for the full details. 

\end{remark}
\subsection{Organization}
The remainder of the paper is organized as follows. In Section~\ref{sec:SD} we review $\sd$ quantization and basic error analysis techniques that will be useful in the rest of the paper. In Section \ref{sec:preliminaries} we introduce the concept of a sub-Gaussian random matrix and recall some of its key properties as well as some important probabilistic tools. In Section \ref{sec:frame}, we prove  a probabilistic lower bound on the singular values of the product of the matrix $D^{-r}$, where $r$ is a positive integer and $D$ is a difference matrix,  and  a sub-Gaussian random matrix. Finally,  we use this result in combination with some known results on the properties of various $\sd$ quantization schemes to prove the main theorems. 

\section{Sigma-Delta quantization}\label{sec:SD}
An $r$th order $\sd$ quantizer $Q^{(r)}: \R^m\mapsto \A^m$ maps a sequence of inputs $(y_i)_{i=1}^m$ to a sequence $(q_i)_{i=1}^m$  whose elements take on values from $\A$ 
via the iteration
\begin{align}\label{eq:sd_iterations}
q_i & = Q\left(\rho(u_{i-1},  \cdots, u_{i-r}, y_i, \cdots, y_{i-r+1})\right)\\
(\Delta^r u)_i & = y_i - q_i. \notag
\end{align}
Here $\rho$ is a fixed function known as the {\em quantization rule} and $(u_i)_{i=1}^m$ is a sequence of state variables initialized to zero, i.e., $u_i=0$ for all $i\leq 0$. It is worth noting that designing a a good quantization rule in the case $r>1$ is generally non-trivial, as one seeks stable $\sd$ schemes, i.e., schemes that satisfy
\begin{equation}\label{eq:stable}
\|y\|_{\infty} \leq C_1 \implies \|u\|_{\infty} \leq C_2,
\end{equation}
for constants $C_1$ and $C_2$ that do not depend on $m$ (note that for the remainder of this paper, the constants are numbered in the order of appearance; this allows to refer to constants introduced in previous results and proofs). In particular, stability is difficult to ensure when one works with a coarse quantizer associated with a small alphabet, the extreme case of which is $1$-bit quantization corresponding to $\A=\{\pm1\}.$

In this work we consider two different sets of assumptions. Our results on compressed sensing reconstruction require sufficiently fine alphabets, whereas the results on frame quantization make no assumptions on the size of the alphabet ---in particular, allowing for very coarse alphabets. 
 In both cases we will work with the $2L$ level {\em mid-rise} alphabet 
\begin{equation}\A=\Big\{\pm (2 j+ 1)\delta/2, \quad j\in\{0,...,L-1\} \Big\}.\label{eq:mid-rise}\end{equation}
\subsection{Greedy sigma-delta schemes}\label{subsec:greedy}
We will work with the {\em greedy} $\sd$ quantization schemes
\begin{align}\label{greedy-quant}
q_i & = Q \Big (\sum_{j=1}^r (-1)^{j-1} {r \choose j} u_{i-j} + y_i  \Big)\\
u_i &=  \sum_{j=1}^r (-1)^{j-1} {r \choose j} u_{i-j} +y_i - q_i.\notag
\end{align}
It is easily seen by induction that for the $2L$ level mid-rise alphabet and $\|y\|_\infty \leq C$, a sufficient condition for stability is $L \geq 2 \lceil\frac{C}{\delta}\rceil+2^r+1$ as this implies 
\begin{equation} \label{greedy-stab}
\|u\|_\infty \leq  \delta/2. 
\end{equation}
Note that to satisfy this stability condition, the number of levels $L$ must increase with $r$. 
\subsection{Coarse sigma-delta schemes} \label{CoSD}
We are also interested in coarse $\sd$ quantization, i.e., schemes where the alphabet size is fixed. In this case, guaranteeing stability with a smaller alphabet  typically entails a worse (i.e., larger) stability constant. The coarse $\sd$
 schemes that we employ were first proposed by G{\"u}nt{\"u}rk \cite{G-exp} and refined by Deift et al. \cite{DGK10}. Originally designed to obtain exponential accuracy in the setting of bandlimited functions, they were used to obtain root-exponential accuracy in the finite frame setup in \cite{KSW12}. At their core is a change of variables of the form $u=g*v$, where $u$ is as in \eqref{eq:sd_iterations} and $g\in \R^{d+1}$ for some $d\geq r$ (with entries indexed by the set $\{0,\hdots,d\}$) such that $g_0=1$. The quantization rule is then chosen in terms of the new variables as $\rho(v_i, v_{i-1},\dots, y_i)=(h*v)_i +y_i$,  where $h=\delta^{(0)}-\Delta^r g$ with $\delta^{(0)}$ the Kronecker delta. Then \eqref{eq:sd_iterations} reads as 
 
\begin{align}\label{eq:SDrec2}
q_i&=Q((h*v)_i+y_i) \\
v_i&=(h*v)_i+y_i-q_i \notag,
\end{align}
where, again, $Q$ is the scalar quantizer associated with the $2L$ level mid-rise alphabet  \eqref{eq:mid-rise}. By induction, one concludes 

$$ \|h\|_1\frac{\delta}{2} + \|y\|_\infty \leq L\delta \implies \|v\|_\infty \leq \frac{\delta}{2} \implies \|u\|_\infty \leq \|g\|_1\frac{\delta}{2}, $$
i.e., a sufficient condition to guarantee stability for all bounded inputs $\|y\|_\infty \leq \mu$ is \begin{equation}\|h\|_1 \leq 2L-\frac{2\mu}{\delta}\label{eq:h_cond}.\end{equation} Thus, one is interested in choosing $g$ with minimal $\|g\|_1$ subject to $h=\delta^{(0)}-\Delta^r g$ and \eqref{eq:h_cond}. This problem was studied in \cite{G-exp, DGK10}
 leading to the following proposition (cf. \cite{KSW12}).
\begin{proposition}\label{prop:best_known}
There exists a universal constant $C_3>0$ such that for any midrise quantization alphabet $\A=\A^\delta_L$, for any order $r\in\N$, and for all $\mu<\delta\left(L-\frac{1}{2}\right)$, there exists $g\in\R^{d+1}$ for some $d\geq r$ such that the $\Sigma\Delta$ scheme given in \eqref{eq:SDrec2} is stable for all input signals $y$ with $\|y\|_\infty\leq \mu$ and
\begin{equation}\label{eq:best_known}
\|u\|_\infty \leq c C_3^r r^r \frac{\delta}{2},
\end{equation}
where $u=g*v$ as above, $C_3=\left(\left\lceil \frac{\pi^2}{(\cosh^{-1} \gamma)^2} \right\rceil \frac{e}{\pi}\right)$ and $\gamma:=2L-\frac{2\mu}{\delta}$.
\end{proposition}

\subsection{Sigma-Delta error analysis}\label{sderranalysis}
As above, assume that $x\in \R^k$ and the frame matrix $E\in\R^{m\times k}$. If the vector of frame coefficients $y=Ex \in \R^m$ is $\sd$ quantized to yield the vector $q\in\A^m$, then linear reconstruction of $x$ from $q$ using some dual frame $F$ of $E$ (i.e., $FE=I$) produces the estimate $\hat{x}:=Fq$. We would like to control the reconstruction error $\eta:=x-\hat{x}$. Writing the state variable equations \eqref{eq:sd_iterations} in vector form, we have
\begin{equation}
D^r u = y - q, 
\end{equation}
 where $D$ is the $m\times m$ difference matrix with entries given by
\begin{equation} D_{ij} = \left\{ 
  \begin{array}{l l}
    1 & \quad i=j\\
-1 &\quad i=j+1\\
  0 & \quad \text{otherwise}
  \end{array} \right..
\end{equation}
Thus,  \begin{equation} \eta= x- Fq = F(y-q) = FD^r u. \end{equation} 
Working with with stable $\sd$ schemes, one can control $\|u\|_2$ via $\|u\|_\infty$. Thus, it remains to bound the operator norm $\|FD^r\|:=\|FD^r\|_{\ell_2^m\mapsto \ell_2^k}$ and a natural choice for $F$ is 
\begin{equation}F:=\arg\min\limits_{G:GE=I}\|GD^r\| = (D^{-r}E)^\dagger D^{-r}.\label{eq:sob_dual}\end{equation}
This so-called Sobolev dual frame was first proposed in \cite{blum:sdf}.
Here $A^\dagger:= (A^*A)^{-1}A^*$ is the $k\times m$ Moore-Penrose (left) inverse of the $m\times k$ matrix $A$.  Since \eqref{eq:sob_dual} implies that $FD^r=(D^{-r}E)^\dagger,$  the singular values of $D^{-r}E$ will play a key role in this paper. 

We begin by presenting some important properties of the matrix $D$.  The following proposition is a quantitative version of Proposition 3.1 of \cite{GLPSY}. 
\begin{proposition}\label{prop:diff_mat}
The singular values of the matrix $D^{-r}$ satisfy 
$$ \frac{1}{(3\pi r)^r}\left(\frac{m}{j}\right)^r\leq \sigma_j(D^{-r})\leq {(6r)^r}\left(\frac{m}{j}\right)^r.$$
\end{proposition}

\begin{proof} Note that (see, e.g., \cite{GLPSY}) $$ \frac{1}{\pi}\left(\frac{m+1/2}{j-1/2}\right)\leq \sigma_j(D^{-1})\leq \frac{1}{2}\left(\frac{m+1/2}{j-1/2}\right).$$ 
Moreover, by Weyl's inequalities \cite{hj} on the singular values of Hermitian matrices, it holds that (see \cite{GLPSY} for the full argument)
$$ \left(\sigma_{\min(j+2r,m)}(D^{-1})\right)^r \leq \sigma_j(D^{-r}) \leq \left(\sigma_{\max(j-2r,1)}(D^{-1})\right)^r.$$
Combining the above inequalities, we obtain 
$$ \frac{1}{\pi^r}\frac{(m+1/2)^r}{\min(j+2r-1/2,m-1/2)^r}\leq \sigma_j(D^{-r})\leq \frac{1}{2^r}\frac{(m+1/2)^r}{\max(j-2r-1/2,1/2)^r}.$$
Observing that $$\min(j+2r-1/2,m-1/2) \leq~j+2r \leq 3rj$$ for $r,j \in \Z^+$ establishes the lower bound. 

For the upper bound, note that $j-2r-1/2 \geq  (2r+1)^{-1} j/2,$ for $j \geq 2r+1$ and  $1/2 \geq (2r+1)^{-1}j/2$ for $j \in \{ 1,..., 2r \}$. Thus, $$\frac{(m+1/2)^r}{2^r \max{(j-2r-1/2,1/2)^r}} \leq \frac{2^r m^r} {2^r (4r+2)^{-r} j^r} \leq (6r)^r \left(\frac{m}{j}\right)^r. $$ \end{proof}

\section{Sub-Gaussian random matrices}\label{sec:preliminaries}
Here and throughout,  $x\sim \mathcal{D}$ denotes that the random variable $x$ is drawn according to a distribution $\mathcal{D}$. Furthermore, $\G(0,\sigma^2)$ denotes the zero-mean Gaussian distribution with variance $\sigma^2$. 
The following definition provides a means to compare the tail decay of two distributions. 
\begin{definition}
If two random variables $\eta\sim\mathcal{D}_1$ and $\xi\sim\mathcal{D}_2$ satisfy $P(|\eta|>t) \leq K P(|\xi|>t)$ for some constant $K$ and all $t\geq0$,  then we say that $\eta$ is $K$-dominated by $\xi$ (or, alternatively, by $\mathcal{D}_2$).  \end{definition}\begin{definition}\label{def:subgvar} A random variable is sub-Gaussian with parameter $c>0$ if it is $e$-dominated by $\G(0,c^2)$.
\end{definition}

\begin{remark}
One can also define sub-Gaussian random variables via their moments or, in case of zero mean, their moment generating functions. See \cite{ve12-1} for a proof that all these definitions are equivalent. 
\end{remark}
\begin{remark}
Examples of sub-Gaussian random variables include Gaussian random variables, all bounded random variables (such as Bernoulli), and their linear combinations. 
\end{remark}

\begin{definition} \label{def:subgmat}
We say that a matrix $E$ is sub-Gaussian with parameter $c$, mean $\mu$ and variance $\sigma^2$ if its entries are independent sub-Gaussian random variables  with mean $\mu$, variance $\sigma^2$, and parameter $c$.\end{definition}
The contraction principle (see, for example, Lemma 4.6 of \cite{LedouxTalagrand91book}) will allow us to derive estimates for sub-Gaussian random variables via the corresponding results for Gaussians. 
\begin{lemma}[Contraction Principle]\label{lem:contraction}
Let $G:\R^+ \mapsto \R^+$ be a non-decreasing convex function and let $(\eta_i)$ and $(\xi_i)$ be two finite symmetric sequences of random variables such that there exists a constant $K\geq 1$ such that for each $i$, $\eta_i$ is $K$-dominated by $\xi_i$.  Then, for any finite sequence $(x_i)$ in a Banach space equipped with a norm $\|\cdot \|$ we have
$$\E G(\|\sum_i \eta_i x_i\|) \leq \E G(K\|\sum_i \xi_i x_i\|).$$ 
\end{lemma}
While the contraction principle as well as the following chaos estimate are formulated for random vectors, we mainly work with random matrices. Thus, it is convenient to ``vectorize" the matrices: for  a matrix $A$, we denote by $\vec{A}$ the vector formed by stacking its columns into a single column vector. 

To state the more refined chaos estimate, we need the concept of the Talagrand $\gamma_2$-functional (see, e.g., \cite{LedouxTalagrand91book} for more details).
\begin{definition} \label{def:gamma-2}
For a metric space $(T,d)$, an {\it admissible sequence} of $T$ is a
collection of subsets of $T$, $\{T_r : r \geq 0\}$, such that for
every $s \geq 1$, $|T_r| \leq 2^{2^r}$ and $|T_0|=1$. The $\gamma_2$ functional is defined by
\begin{equation*}
\gamma_2(T,d) =\inf \sup_{t \in T} \sum_{r=0}^\infty
2^{r/2}d(t,T_r),
\end{equation*}
where the infimum is taken with respect to all admissible sequences
of $T$.
\end{definition}
Furthermore, for $\mathcal A$ a set of matrices, we denote by $d_{Fr}(\mathcal A):=\sup\limits_{A\in\mathcal A} \|A\|_{Fr}$ and $d_{\ell_2\rightarrow \ell_2}(\mathcal A):= \sup\limits_{A\in\mathcal A}\|A\|_{\ell_2\rightarrow \ell_2}$ the diameter in the Frobenius norm $\|\cdot\|_{Fr}$ and the spectral norm $\|\cdot\|_{\ell_2\rightarrow \ell_2}$, respectively. Here the Frobenius norm  is given by $\|A\|_{Fr}=\|\vec A\|_{2}$ and  the spectral norm is given by $\|A\|_{\ell_2\rightarrow \ell_2}=\sup\limits_{\|x\|_{2}=1} \|Ax\|_{2}$. The following theorem is a slightly less general version of \cite[Thm. 3.1]{krmera12}.

\begin{theorem} \label{thm:main-tail}
Let ${\mathcal A}$ be a symmetric set of matrices, that is, $\mathcal A=-\mathcal A$, and let $\xi$ be a random vector whose entries $\xi_j$ are independent, sub-Gaussian random variables of mean zero, variance one, and parameter $c$. Set

\begin{align*}
\mu &=\gamma_2({\mathcal A},\| \cdot \|_{\ell_2 \to \ell_2}) \left( \gamma_2({\mathcal A},\| \cdot \|_{\ell_2 \to \ell_2})
+d_{Fr}({\mathcal A})\right),\notag\\
\nu_1 &=d_{\ell_2 \to \ell_2}({\mathcal A})(\gamma_2({\mathcal A},\| \cdot \|_{\ell_2 \to \ell_2})+d_{Fr}({\mathcal A})), \ \ {\rm and} \ \ \nu_2=d_{\ell_2 \to \ell_2}^2({\mathcal A}).
\end{align*}
Then, for $t>0$,
\begin{equation*}
\P \left( \sup_{A \in {\mathcal A}} \left| \|A {\xi}\|_2^2 - \E \|A\xi\|_2^2 \right|\geq C_4 \mu + t \right) \leq 2 \exp\left(-C_5 \min\left\{\frac{t^2}{\nu_1^2},\frac{t}{\nu_2}\right\}\right).
\end{equation*}
The constants $C_4, C_5$ depend only on $c$. 
\end{theorem}

%

\section{Main results}
\label{sec:frame}

\subsection{Estimates of singular values and operator norms}

As argued above, a key quantity to control the $\sd$ reconstruction error both in the context of compressed sensing and frame quantization is the norm $\|FD^r u\|_2$, where  $F$ is the Sobolev dual of a sub-Gaussian frame $E$ and $u$ is the associated state vector. This quantity can be controlled by the product of the operator norm $\|FD^r\| = \frac{1}{\sigma_{\rm{min}}(D^{-r}E)}$ and the vector norm $\|u\|_2$. We will estimate these quantities separately in the following two propositions. An estimate for the first quantity can be deduced from the following proposition together with two observations:  First, recall that singular values are invariant under unitary transformations, so $\sigma_{\rm{min}}(D^{-r}E)=\sigma_{\rm{min}}(SV^*E)$, where $USV^*$ is the singular value decomposition of $D^{-r}$. Second, when estimated using Proposition~\ref{prop:diff_mat}, the singular values of $D^{-r}$ are bounded exactly as in the following assumptions.

\begin{proposition}
\label{thm:main_1} 
Let $E$ be an $m \times k$ sub-Gaussian matrix with mean zero, unit variance, and parameter $c$, let $S=\operatorname{diag}(s)$ be a diagonal matrix, and let $V$ be an orthonormal matrix, both of size $m \times m$. Further, let $r\in \Z^+$ and suppose that $s_j\geq C_6^r \left(\frac{m}{j}\right)^r$, where $C_6$ is a positive constant that may depend on $r$. Then there exist constants $C_7, C_8 >0$ (depending on $c$ and $C_6$) such that for $0<\alpha<1$ and $\lambda:= \frac{m}{k} \geq C_7^{\frac{1}{1-\alpha}}$

$$\P\left( \sigma_{\min}(\tfrac{1}{\sqrt{m}} SV^*E) \leq  \lambda^{\alpha(r-1/2)}  \right) \leq 2\exp( - C_8 m^{1-\alpha} k^{\alpha}).$$
In particular, $C_8$ depends only on $c$, while $C_7$ can be expressed as ${f(c)} { C_6^{-\frac{2r}{2r-1} }}$ provided $C_6 \leq 1/2$.
\end{proposition}

\begin{proof}
The matrix $SV^*E$ has dimensions $m$ and $k$, so by the Courant min-max principle applied to the transpose one has
\begin{align}
  \sigma_{\min}(SV^*E) &= \min\limits_{\substack{W\subset \R^m\\ \dim W= m-k+1}} \sup\limits_{ z\in W: \|z\|_2=1 }\|E^*V S z\|_2 
\end{align}
Noting that, for $m\geq\widetilde k :=C_9 m^{1-\alpha} k^\alpha>k$, where the constant $C_9$ will be determined later, each $m-k+1$-dimensional subspace intersects the span $V_{\widetilde k}$ of the first $\widetilde k$ standard basis vectors in at least a $\widetilde k-k+1$-dimensional space, this expression is bounded from below by 
\begin{align}
 \min\limits_{\substack{W\subset V_{\widetilde k}\\ \dim W= \widetilde k-k+1}} \sup\limits_{ z\in W: \|z\|_2=1 }\|E^*V S z\|_2 
 &\geq \min\limits_{\substack{W\subset V_{\widetilde k}\\ \dim W= \widetilde k-k+1}} \sup\limits_{ z\in W: \|z\|_2=s_{\widetilde k} }\|E^*V z\|_2\\
 &=\min\limits_{\substack{W\subset \R^{\widetilde k}\\ \dim W= \widetilde k-k+1}} \sup\limits_{ z\in W:  \|z\|_2=1}s_{\widetilde k}\|E^* V P^*_{\widetilde k}z\|_2. \label{eq:svdual}
\end{align}
The inequality follows from the observation that $V_{\widetilde k}$ is invariant under $S$ and the smallest singular value of $S|_{V_{\widetilde k}}$ is $s_{\widetilde k}$. In the last step, $P_{\widetilde k}\in \R^{\widetilde k\times m}$ denotes the projection of an $m$-dimensional vector onto its first $\widetilde k$ components. We note that \eqref{eq:svdual}, again by the Courant min-max principle, is equal to
\begin{align}
 s_{\widetilde k}\sigma_k (E^* V P^*_{\widetilde k}) &= s_{\widetilde k}\sigma_{\rm{min}} (P_{\widetilde k} V^* E) =s_{\widetilde k} \inf_{y\in S^{k-1}} \|P_{\widetilde k} V^* E y\|_2 
 \end{align}
Now, as $\E\|P_{\widetilde k} V^* E y\|_2^2=\widetilde k$, 
 \begin{align}
\inf_{y\in S^{k-1}} \|P_{\widetilde k} V^* E y\|_2^2  \geq \Big({\widetilde k}-\sup _{y\in S^{k-1}} \big|\|P_{\widetilde k} V^* E y\|^2_2-\E\|P_{\widetilde k} V^* E y\|_2^2\big|\Big). 
\end{align}
Thus, noting that $\frac{\lambda^{\alpha(r-1/2)}}{s_{\widetilde k}} < m^{\alpha(r-1/2)} k^{-\alpha(r-1/2)} C_6^{-r}  m^{-r} {\widetilde k}^r = C_6^{-r} C_9^{r-\frac{1}{2}} \frac{\sqrt{\widetilde k}}{\sqrt{m}}$ and that by choosing $C_9=\min(\tfrac{1}{2}C_6^\frac{2r}{2r-1}, \tfrac{1}{2})$ we ensure that $1-C_6^{-2r} C_9^{2r-1 }\geq \frac{1}{2}$,
\begin{align}
  \P(\sigma_{\min}&(\tfrac{1}{\sqrt{m}} SV^*E) \leq \lambda^{\alpha(r-1/2)}) \\ &\leq \P(\sup _{y\in S^{k-1}} \big|\|\tfrac{1}{\sqrt{m}} P_{\widetilde k} V^* E y\|^2_2-\E\|\tfrac{1}{\sqrt{m}}P_{\widetilde k} V^* E y\|_2^2\big|
  \geq (1-C_6^{-2r} C_9^{2r-1 })\frac{{\widetilde k}}{ m })\\
  &\leq\P(\sup _{y\in S^{k-1}} \big|\|\tfrac{1}{\sqrt{m}} P_{\widetilde k} V^* E y\|^2_2-\E\|\tfrac{1}{\sqrt{m}}P_{\widetilde k} V^* E y\|_2^2\big|
  \geq \frac{{\widetilde k}}{ 2m }). \label{eq:conc}
\end{align}
Note that this choice of $C_9$ also ensures $\widetilde k\leq m$, which is required above.
We will estimate \eqref{eq:conc} using Theorem~\ref{thm:main-tail}, similarly to the proof of \cite[Thm. A.1]{krmera12}.
Indeed, we can write
\begin{equation}
 \tfrac{1}{\sqrt{m}}P_{\widetilde k} V^* E y =  W_y \xi,
\end{equation}
where $\xi=\overrightarrow{E^*}$ is a vector of length $km$ with independent sub-Gaussian entries of mean zero and unit variance, and
\begin{equation}
 W_y = \frac{1}{\sqrt{m}} P_{\widetilde k} V^* \left( \begin{matrix} y^T & 0 & \cdots & 0\\
0 & y^T & \cdots & 0\\
\vdots & \vdots & \vdots & \vdots\\
0 & \cdots & 0 & y^T
\end{matrix} \right).
\end{equation}
In order to apply Theorem~\ref{thm:main-tail}, we need to estimate, for $\mathcal A = \{W_y : y\in S^{k-1}\}$, $d_{Fr}(\mathcal A)$, $d_{\ell_2\rightarrow \ell_2}(\mathcal A)$, and $\gamma_2({\mathcal A}, \|\cdot\|_{\ell_2\rightarrow \ell_2})$. We obtain for $A= W_y\in\mathcal A$:
\begin{equation}
 \|A\|_{Fr}^2= \frac{1}{m}\sum_{j=1}^k \sum_{p_1,p_2=1}^{\widetilde k, m} y_j^2 V_{p_1,p_2}^2 =\frac{\widetilde k}{m}, \quad \text{ so } \quad d_{Fr}(\mathcal A) = \sqrt{\frac{\widetilde k}{m}}.
\end{equation}
Furthermore, we have, for $z\in \R^k$,
\begin{equation}
\| W_z\|_{\ell_2\rightarrow \ell_2} = \left\|\frac{1}{\sqrt{m}} P_{\widetilde k} V^* \left( \begin{matrix} z^T & 0 & \cdots & 0\\
0 & z^T & \cdots & 0\\
\vdots & \vdots & \vdots & \vdots\\
0 & \cdots & 0 & z^T
\end{matrix} \right)\right\|_{\ell_2\rightarrow \ell_2}\leq \left\|\frac{1}{\sqrt{m}}\left( \begin{matrix} z^T & 0 & \cdots & 0\\
0 & z^T & \cdots & 0\\
\vdots & \vdots & \vdots & \vdots\\
0 & \cdots & 0 & z^T
\end{matrix} \right)\right\|_{\ell_2\rightarrow \ell_2},
\end{equation}
so the quantities $d_{\ell_2\rightarrow \ell_2}(\mathcal A)$ and $\gamma_2({\mathcal A}, \|\cdot\|_{\ell_2\rightarrow \ell_2})$ can be estimated in exact analogy to  \cite[Thm. A.1]{krmera12}.
This yields $d_{\ell_2\rightarrow \ell_2}({\mathcal A}) = \frac{1}{\sqrt{m}}$ and $\gamma_2({\mathcal A}, \|\cdot\|_{\ell_2\rightarrow \ell_2}) \leq C_{10} \sqrt{\frac{k}{m}}$ for some constant $C_{10}$ depending only on $c$.
With these estimates, we obtain for the quantities $\mu$, $\nu_1$, $\nu_2$ in Theorem~\ref{thm:main-tail}
\begin{align}
 \mu\leq & (2C_{10}+2) \frac{\sqrt{k\widetilde k}}{m}\\
 \nu_1\leq & (C_{10} +1)\frac{\sqrt{\widetilde k}}{m}\\
 \nu_2\leq & \frac{1}{m},
\end{align}
so the resulting tail bound reads
\begin{align}
 \P&(\sup _{y\in S^{k-1}} \big|\|\tfrac{1}{\sqrt{m}}P_{\widetilde k} V^* E y\|^2_2-\E\|\tfrac{1}{\sqrt{m}}P_{\widetilde k} V^* E y\|_2^2\big|\geq C_4 (2C_{10}+2)\frac{\sqrt{k\widetilde k}}{m} + t) \leq e^{- C_5 \min \big(\frac{t^2 m^2}{(C_{10} + 1) \widetilde k} , mt\big)}. 
\end{align}
where $C_4$ and $C_5$ are the constants depending only on $c$ as they appear in Theorem~\ref{thm:main-tail}.
Note that $k=\widetilde k \frac{\lambda^{-(1-\alpha)}}{C_9}$, so for oversampling rates $\lambda>\big( (4C_4  (2C_{10}+2))^2/C_9\big)^{\frac{1}{1-\alpha}}=: C_7^{\frac{1}{1-\alpha}}$, we obtain $C_4 E\leq \frac{\widetilde k}{4m}$ and hence, choosing $t=\frac{\widetilde k}{4m}$, we obtain the result
\begin{equation}
  \P(\sup _{y\in S^{k-1}} \big|\|\tfrac{1}{\sqrt{m}}P_{\widetilde k} V^* E y\|^2_2-\E\|\tfrac{1}{\sqrt{m}}P_{\widetilde k} V^* E y\|_2^2\big|\geq \frac{\widetilde k}{2m}) \leq e^{- C_8 \widetilde k} 
\end{equation}
where, as desired, the constant $C_8:= \frac{C_5}{16 (C_{10}+1)}$ depends only on the sub-Gaussian parameter $c$.
\end{proof}

In contrast to the term $\|FD^r\|$ analyzed in the previous proposition, $\|u\|_2$ crucially depends on the quantization procedure that is employed. The procedure employed will be fundamentally different in the frameworks of compressed sensing and frame quantization. While the quantization level in the compressed sensing scheme is chosen sufficiently fine to allow for accurate support recovery via standard compressed sensing techniques, there is no need for this in the context of frame quantization and the quantization scheme employed can be coarse.

In both cases, we will employ $\sd$ schemes which are stable in the sense of \eqref{eq:stable}. 
 As explained in Section \ref{sec:SD}, $\|u\|_2$ can be controlled for such schemes via the input $Ex$. More precisely, to bound $\|u\|_2$, we require a bound on $\|Ex\|_\infty \leq \|E\|_{\ell_2\rightarrow\ell_\infty} \|x\|_2$. Since the matrices $E=E(\omega)$ are random we derive bounds on the operator norms $\|E\|_{\ell_2\rightarrow\ell_\infty}$ that hold with high probability on the draw of $E$.

\begin{proposition}\label{prop:infty_norm_bound}
 Let $\widetilde{E}$ be an $m \times k$ sub-Gaussian matrix with mean zero, unit variance, and parameter $c$, let $E= \frac{1}{\sqrt{m}}\widetilde{E}$ and fix $\alpha\in (0,1)$. Denote the associated oversampling rate by $\lambda:=m/k$.   
Then, with probability at least $1-e^{-\frac{1}{4} \lambda^{1-\alpha}k}$, we have  for all $\lambda> C_{11}^{\frac{1}{1-\alpha}}$
\begin{equation}\|E\|_{\ell_2\rightarrow\ell_\infty} \leq  e^{1/2}\lambda^{-\frac{\alpha}{2}}.\label{eq:infty_norm_bound}\end{equation}
Here $C_{11}$ is a constant that may depend on $c$, but that is independent of $k$ and $\alpha$.
\end{proposition}
\begin{proof}
Since
 \begin{equation}
 \|E\|_{\ell_2 \rightarrow \ell_\infty} = \max_{i\in\{1,...,m\}}  \Big(\sum_{j=1}^{k} e_{ij}^2\Big)^{1/2}=\max_{\eta \text{ is a row of E}} \|\eta\|_2,
 \end{equation}
 we will focus on bounding the norm $\|\eta\|_2$ for random vectors $\eta \in \R^k$ consisting of independent sub-Gaussian entries $\eta_i$ with parameter $\frac{1}{\sqrt{m}}$. Using Markov's inequality as well as the contraction principle applied to the increasing convex function $G: x\mapsto e^{x^2t/2}$ and the sequence $x_i=e_i$ of standard basis vectors, we reduce to the case of a $k$-dimensional random vector $\xi \sim \G(0,\frac{1}{m} I)$:
 \begin{align}
  \P\Big(\|\eta\|^2_2   \geq \Theta/ \lambda\Big)&\leq \inf_{t>0} e^{-t\Theta/(2\lambda)}\E e^{\|\eta\|_2^2t/2} \\
  &\leq \inf_{t>0} e^{-t\Theta/(2\lambda)}\E e^{e\|\xi\|_2^2t/2}\\
      &= \inf_{t>0} e^{-t\Theta/(2\lambda)} \big(1-et/m\big)^{-k/2}\\
 &\leq \Big(\frac{\Theta}{e}\Big)^{k/2} \exp\Big(-(\Theta/e-1)k/2 \Big)
\\&= \Theta^{k/2} \exp\Big(-\frac{\Theta}{e}\frac{k}{2} \Big),
 \end{align}
 where we set $t=m(\frac{1}{e}-\frac{1}{\Theta})$ to obtain the third inequality.
 Applying a union bound over the $m$ rows and specifying $\Theta=e\lambda^{1-\alpha}$ we obtain for $\lambda$ sufficiently large:
 \begin{align}
\P( \|E\|_{\ell_2 \rightarrow \ell_\infty} \geq e^{1/2}\lambda^{-\alpha/2}) 
&\leq m  e^{k/2}\lambda^{(1-\alpha) k/2} \exp\Big(-\lambda^{1-\alpha}k/2 \Big)\\
&= k\lambda \left(\lambda^{1-\alpha} \exp\Big(-\frac{1}{2}\lambda^{1-\alpha} \Big)\right)^{k/2} e^{-\frac{1}{4} m^{1-\alpha}k^\alpha}\\
&=: k\lambda f(\lambda^{1-\alpha})^{-k/2} e^{-\frac{1}{4} m^{1-\alpha}k^\alpha}\\
&\leq e^{-\frac{1}{4} m^{1-\alpha}k^\alpha},
 \end{align}
 where we used that $f$ is independent of $k$ and grows superlinearly, so above some threshold $C_{11}^{\frac{1}{1-\alpha}}$, $f(\lambda^{1-\alpha})^{-k/2}$ can absorb both $\lambda$ and $k$.

\end{proof}

\begin{remark}
Clearly, when the entries of $\widetilde E$ are \emph{bounded} random variables, there exists a finite, deterministic upper bound on the operator norm $\|E\|_{\ell_2\rightarrow\ell_\infty}$, obtained via the bounds on the matrix entries. In fact, for Bernoulli matrices, the resulting 
bounds are sharp.
\end{remark}

\subsection{Root-exponential accuracy for sub-Gaussian frames}
\begin{theorem}\label{thm:root-exponential}
Let $\widetilde{E}$ be an $m\times k$ sub-Gaussian  
matrix with mean zero, unit variance, and parameter $c$,  let $E=\frac{1}{\sqrt{m}}\widetilde{E}$ and suppose that 
\[\lambda:=\frac{m}{k} \geq C_{12},    \]
where $C_{12}$ is an appropriate constant that only depends on $c$.
For a vector $x\in \R^k$, denote by $Q^{(r^*)}(Ex)$ the $2L$-level $\sd$ quantization of $Ex$ using a scheme as given in \eqref{eq:SDrec2} which satisfies \eqref{eq:best_known}  with order $r^*:= \lfloor {\lambda^{1/2}}/(2C_{13})\rfloor $ and $\delta> \frac{4 e^{1/2}}{\lambda^{1/4}L}$. Denote by $F$ the $r^*$-th order Sobolev dual of $E$. Then, with probability exceeding $1-3e^{-C_{14} \sqrt{mk}}$ on the draw of $E$, the reconstruction error satisfies 
$$\|x-FQ^{(r)}(Ex)\|_2  \leq C_{15} \sqrt{k} e^{-C_{13} {\lambda^{1/2}}}\delta$$
uniformly for all $x$ with $\|x\|_2\leq 1$, 
where $C_{13}$, $C_{14}$, and $C_{15}$ are appropriate constants depending only on $c$. 
If the entries of $\widetilde E$ are bounded by an absolute constant $K$ almost surely, the condition on $\delta$ can be relaxed to $\delta>\frac{4 K}{L}\lambda^{-1/2}$.
\end{theorem}

\begin{proof}
Observe the following facts:
\begin{enumerate}
\item[(I)] Since $FD^r = (D^{-r}E)^\dagger$, then
$\|FD^r\| = \frac{1}{\sigma_{min}(D^{-r}E)}=\frac{1}{\sigma_{min}(SV^*E)}$, where $D^{-r} = USV^*$ is the singular value decomposition of $D^{-r}$. Thus, by Proposition \ref{thm:main_1} with $C_6 = 1/(3\pi r)$ and $\alpha=1/2$, it holds that with probability greater than $1-2e^{-C_8 \sqrt{mk}}$,  
\[\|FD^r\| \leq \lambda^{- \frac{1}{2}(r-1/2)}\]
provided $\lambda\geq C_{16}^2 (3\pi r)^\frac{4r}{(2r-1)}.$ Since $(3\pi r)^\frac{1}{{2r-1}}$ is decreasing with $r \geq1$, 
the condition is satisfied if \begin{equation}\lambda^{1/2}\geq C_{16}(3\pi)^2 r =: C_{17}r. \label{eq:lambda_cond}\end{equation}
Without loss of generality, we may assume $C_{17} \geq 1$. 
\item[(II)] By Proposition \ref{prop:infty_norm_bound} with $\alpha=1/2$ and for $\lambda>C_{11}^2$, we have   $$\|E\|_{\ell_2 \rightarrow \ell_\infty} \leq e^{1/2} \lambda^{-1/4}$$ with probability greater than $1-e^{-\frac{1}{4} \sqrt{mk}}$.
If the entries of $E$ are bounded by $K/\sqrt{m}$ a.s. then  \[\|E\|_{\ell_2 \rightarrow \ell_\infty} =\max_{\eta \text{ is a row of E}} \|\eta\|_2 \leq K\lambda^{-1/2} \quad \text{ almost surely.}\]
\item[(III)] As this implies that the quantizer input satisfies $\|Ex\|_\infty \leq e^{1/2} \lambda^{-1/4} \|x\|_2$ or $\|Ex\|_\infty \leq K \lambda^{-1/2} \|x\|_2$, respectively, our assumption \eqref{eq:best_known} implies that for all signals $x \in\R^k$ with $\|x\|_2\leq 1$, the resulting state vector satisfies $$\|u\|_\infty \leq C_{18} C_{19}^r r^r \delta/2,$$ where $C_{19} = \lceil \frac{\pi^2}{(\cosh^{-1}\gamma)^2}\rceil \frac{e}{\pi}$ for $\gamma = 2L - 2 e^{1/2}\lambda^{-1/4}/\delta$ or $\gamma = 2L -2K \lambda^{-1/2}/\delta$.
By Cauchy-Schwarz, this entails \[\|u\|_2\leq C_{18} C_{19}^r r^r \delta/2 \sqrt{m}.\]
The respective constraints on $\delta $ ensure that in both cases, $\frac{3}{2}L<\gamma<2L$, so $C_{19}$ is well-defined and bounded above by $10$. Moreover, $\gamma$ and $C_{19}$ effectively do not depend on $\lambda$. 

\end{enumerate}
Thus, the reconstruction error $x-FQ^{(r)}(Ex) = FD^r u$ satisfies 
$$\|FD^r u\|_2 \leq C_{18} C_{19}^r r^{r} \lambda^{-\frac{1}{2}(r-1/2) }\frac{\delta}{2}\sqrt{m} =:f(r). $$ 
Motivated by the observation \[\arg\min\limits_{r}f(r) = e^{-1}\frac{{\lambda^{1/2}}}{{ C_{19}}}\]  we define $r^*= 
\Big\lfloor  \frac{e^{-1}\lambda^{1/2}}{C_{17}  C_5}  \Big\rfloor, $ which is adjusted to satisfy \eqref{eq:lambda_cond}. As the above argument only works for $r\geq 1$, we need the additional requirement that $\lambda \geq C_{20}:=100 e^2 C_{17}^2$. We estimate for such $\lambda$
 \begin{align}f( r^* ) &\leq  C_4 e \cdot {m^{1/2}\lambda^{1/4}} (C_{17} e)^ { - \frac{1}{e}{\frac{\lambda^{1/2}}{ C_{17}C_5 }}  }\delta/2 \leq C_4 e\sqrt{k}\lambda^{3/4} \exp \left( - \frac{1}{e}{\frac{\lambda^{1/2}}{ C_{17} C_5 }}  \right)\delta/2  &
\\ &\leq C_{21}\sqrt{k} \lambda^\frac{3}{4}e^{-2C_{13}{\lambda^{1/2}}}{\delta}, &
\\ &\leq C_{22}\sqrt{k} e^{-C_{13}{\lambda^{1/2}}}{\delta},
\end{align}
where $C_{13}$, $C_{21}$, and $C_{22}$ are appropriate constants.
Here the last inequality stems from the observation that the polynomial factor is dominated by the exponential term.  
The theorem follows by choosing $C_{12} = \max\{C_{11}^2, C_{20}\}$.
\end{proof}

By combining slight modifications of (I), (II), (III) in the proof of Theorem \ref{thm:root-exponential} above, one can deduce the following proposition which will be useful when dealing with quantization in the compressed sensing context. 
Below $F$ is $r$-th order Sobolev dual of $E$ and $Q^{(r)}(Ex)$ is the $r$-th order, $2L$-level greedy $\sd$ scheme, with step-size  $\delta$  as in Section \ref{subsec:greedy}. 

\begin{proposition}\label{prop:polynomial}
Let $\widetilde{E}$ be an $m\times k$ sub-Gaussian  
matrix with mean zero, unit variance, and parameter $c$, let $E=\frac{1}{\sqrt{m}}\widetilde{E}$, and fix an integer $r\geq 1$. Suppose that 
\[\lambda:=\frac{m}{k} \geq C_{12}',    \]
where $C_{12}'$ is an appropriate constant that only depends on $c$ and $r$. Let $0<\alpha<1$.
For an integer $L\geq \lceil\frac{e^{1/2}\lambda^{-\alpha/2}}{\delta}\rceil +2^r+1$, 
 with probability exceeding $1-3e^{-C_{14}' m^{1-\alpha}k^\alpha}$ on the draw of $E$, the reconstruction error satisfies 
$$\|x-FQ^{(r)}(Ex)\|_2  \leq  \lambda^{-\frac{\alpha}{2}(r-1/2) }\frac{\delta}{2}\sqrt{m} $$
uniformly for all $x$ with $\|x\|_2\leq 1$, 
where  $C_{14}'$, is an appropriate constant depending only on $c$. 
If the entries of $\widetilde E$ are bounded by an absolute constant $K$ almost surely, the condition on $\delta$ can be relaxed to $L\geq \lceil\frac{K\lambda^{-1/2}}{\delta}\rceil +2^r+1$.
\end{proposition}
\begin{proof}
The proof proceeds along the same lines as the proof of Thm 17. In fact, parts (I) and (II) are identical, albeit with a general $\alpha$ rather than $\alpha=\frac{1}{2}$. In part (III) we obtain that the state vector satisfies $$\|u\|_\infty \leq  \delta/2,$$ which entails \[\|u\|_2\leq \frac{\delta}{2} \sqrt{m}.\]
Combining (I) and (III) we see that $$\|FD^ru\|_2 =\|x-FQ^{(r)}(Ex)\|_2  \leq  \lambda^{-\frac{\alpha}{2}(r-1/2) }\frac{\delta}{2}\sqrt{m}. $$
Since the involved constants only depend on $c$, the statement follows.
\end{proof}

\begin{remark}
Note that in Theorem~\ref{thm:root-exponential}, the normalization of the entries depends on the number of measurements taken. In particular, this entails that the number of measurements needs to be assigned  a-priori. In practice, however, one may wish to acquire as many measurements as possible and quantize the measurements on the fly. This can be accomplished when the entries of $\widetilde E$ are bounded by some constant $K$. In this case, $E$ consists of entries bounded by $K/\sqrt{m}$, so Theorem~\ref{thm:root-exponential} predicts root-exponential accuracy for all quantization levels $\delta>2K \lambda^{-1/2}$. Hence the quantization level can be chosen independently of $m$ even if the whole system is multiplied by $\sqrt{\lambda}$. This multiplication yields normalized rows rather than columns, so the normalization does not depend on $m$, as desired.
\end{remark}

\subsection{Polynomial accuracy for compressed sensing quantization}
Analogously to \cite{GLPSY}, the above bounds on the singular values
can be used to establish recovery guarantees from $\sd$ quantized
compressed sensing measurements. To obatin these guarantees, as in
\cite{GLPSY}, the decoding is performed via a two-stage algorithm:
First, a \emph{robust} compressed sensing decoder is used to recover a
coarse estimate and particularly the support of the underlying sparse
signal from the quantized measurements. Next, the estimate is refined
using a Sobolev-dual based reconstruction method. 


\begin{definition}
  The restricted isometry constant (see, e.g., \cite{CRT05})
  $\gamma_k$ of a matrix $\Phi \in \R^{m\times N}$ is the smallest
  constant for which
$$(1-\gamma_k)\|x\|_2^2\leq \|\Phi x\|_2^2 \leq (1+\gamma_k )\|x\|_2^2
$$
for all $x\in \Sigma_k^N$.
\end{definition}

\begin{definition}
  Let $\varepsilon>0$, let $m,N$ be positive integers such that
  $m<N$ and suppose that $\Phi\in\R^{m\times N}$. We say that
  $\Delta: \R^{m\times N}\times\R^m\rightarrow \R^N$ is a robust compressed sensing decoder with
  parameters $(k, a,\gamma)$, $k<m$, and constant $C$ if
\begin{equation}
\|x-\Delta(\Phi, \Phi x+e)\| \le C \varepsilon, 
\end{equation}
for all $x\in \Sigma_k^N$, $\|e\|_2 \le \varepsilon$, and all matrices
$\Phi$ with a restricted isometry constant
$\gamma_{ak}<\gamma$. 
\end{definition}
\begin{remark}
 Examples for robust compressed sensing decoders include $\ell_1$-minimization \cite{CRT05}, greedy algorithms such as orthogonal matching pursuit \cite{zh11}, and greedy-type algorithms such as CoSaMP \cite{netr08}.
\end{remark}

One then obtains the following result.
\begin{theorem}\label{thm:polynomial}
Let $r\in \Z^+$, fix $a\in\N$, $\gamma<1$, and $c, C>0$. Then there exist constants $C_{23}, C_{24},
C_{25}, C_{26}$ depending only on these parameters such that the following holds.

Fix 
$0<\alpha<1$, let $\Phi$ be an $m \times N$ sub-Gaussian matrix with mean zero, unit
variance, and parameter $c$, and suppose that $k\in\N$ satisfies 
  $$\lambda:=\frac{m}{k}\geq \Big(C_{23}  \log(eN/k)\Big)^{\frac{1}{1-\alpha}}.$$

Furthermore, suppose that $\Delta:\R^{m\times N}\times \R^m\rightarrow \R^N$ is a robust compressed sensing
decoder with parameters $(k,a,\gamma)$ and constant $C$.

 Consider the $2L$-level $r$th order greedy $\sd$ schemes with step-size $\delta$, $L\geq \lceil\frac{K\lambda^{-1/2}}{\delta}\rceil +2^r+1$. Denote by $q$ the quantization output resulting from $\Phi z$ where $z\in\R^N$. Then with probability exceeding $1-4e^{-C_{24} m^{1-\alpha}k^{\alpha}}$ for all $z\in \Sigma_k^{N}$ having $\min\limits_{j\in \rm{supp}{(z)}}|z_j| > C_{25}  \delta$:
 \be
 \item[(i)] The support of $z$, $T$, coincides with the support of the best $k$-term approximation of $\Delta(q)$. 
 \item[(ii)] Denoting by $E$ and $F$ the sub-matrix of $\Phi$ corresponding to the support of $z$ and its $r$th order Sobolev dual respectively, and by $x\in\R^k$ the restriction of $z$ to its support, we have
  $$\|x-Fq\|_2 \leq C_{26} \lambda^{-\alpha (r-1/2)}\delta.$$ 
 \ee 
\end{theorem}

The proof essentially traces the same steps as in \cite{GLPSY}.  In
particular, as $\Phi$ is a sub-Gaussian matrix, and by choosing $C_{23}$ large enough, we can ensure that $m \geq
2C_{27}\gamma^{-2}ak \ln\frac{eN}{ak}$ and obtain $\gamma_{ak}\leq \gamma$
with probability exceeding $1-2e^{-\gamma^2m/(2C_{27})}$ (see, e.g.,
\cite[Theorem 9.2]{HolgerBook}). Here $C_{27}$ depends only $c$. Since
$\Delta$ is a robust compressed sensing decoder, (i) follows from
\cite[Proposition 4.1]{GLPSY}.

 Once the support of $x$, say $T$, is recovered,
Proposition~\ref{prop:polynomial} applies to the $m\times k$ sub-Gaussian
matrix $\Phi_T$, the submatrix of $\Phi$ consisting of columns of
$\Phi$ indexed by $T$. Finally, (ii) follows from a union bound over
all submatrices consisting of $k$ columns of $\Phi$. This union bound, coupled with the probability bound used to obtain a restricted isometry constant $\gamma_{ak}<\gamma$ yields the probability in the statement of the theorem as well as the condition on $\lambda$. As all the
proof ingredients established above follow closely the
corresponding results in \cite{GLPSY}, we omit the details.

\begin{remark}
 Note that root-exponential error decay, using quantization with a fixed number of levels, as obtained for the frame quantization case in Theorem~\ref{thm:root-exponential} is hindered in the compressed sensing scenario by the support recovery step of the two stage algorithm in \cite{GLPSY}. 
\end{remark}

\section*{Funding}
This work was supported in part by a Banting Postdoctoral Fellowship administered by the Natural Sciences and Engineering Research Counsel of Canada (NSERC) [to R.S.];   an NSERC Discovery Grant [to {\"O.Y.}];  an NSERC Discovery Accelerator Supplement Award [to {\"O.Y.}].

\section*{Acknowledgment}
The authors thank Evan Chou for useful conversations and comments that greatly improved this manuscript.

\bibliographystyle{IEEEtran}
\bibliography{quantization}

\end{document}